\documentclass{article}
\usepackage[utf8]{inputenc}
\usepackage[american]{babel}
\usepackage{enumerate}

\usepackage{amsfonts,amsmath,amssymb,amsthm,amstext,latexsym}	

\usepackage{hyperref}
\usepackage{xspace}
\usepackage{array,multirow,makecell}
\usepackage{todonotes}
\usepackage{ifthen}
\usepackage{intcalc}
\usepackage[ruled,vlined,boxed,linesnumbered,commentsnumbered]{algorithm2e}
\usepackage{subfig}
\usepackage{paralist}

\usepackage{tikz}
\usetikzlibrary{calc}
\usetikzlibrary{arrows,shapes}
\usetikzlibrary{patterns,snakes}

\makeatletter
\newcommand{\algorithmfootnote}[2][\scriptsize]{%
  \let\old@algocf@finish\@algocf@finish
  \def\@algocf@finish{\old@algocf@finish
    \leavevmode\rlap{\begin{minipage}{\linewidth}
    #1#2
    \end{minipage}}%
  }%
}
\makeatother

\theoremstyle{plain}
\newtheorem{lemma}{Lemma}
\newtheorem{theorem}{Theorem}

\newtheorem{corollary}{Corollary}

\theoremstyle{definition}
\newtheorem{definition}{Definition}

\newcommand{\ttrdy}{-1.5} 
\newcommand{\ttrdx}{-15} 
\newcommand{\xlegend}[2]{
\draw[-, very thick] ($(#1, \ttrdy +1)$) -- ($(#1, \ttrdy -1)$) node [pos=1, below] {$\scriptstyle{#2}$};
}
\newcommand{\ylegend}[2]{
\draw[-, very thick] ($(\ttrdx +1,#1)$) -- ($(\ttrdx -1,#1)$) node [pos=1, left] {$\scriptstyle{#2}$};
}

\newcommand{\LRect}[5]{%
    \draw [black,fill=white]  (#1,#2) -- ( #3, #2) -- ( #3,#4) -- (#1,#4);
    \node () at ($(#1 /2 + #3 /2, #2 /2 + #4 /2)$) {$\scriptstyle{#5}$};
 }

\newcommand{\RRect}[5]{%
    \draw [black,fill=white]  ( #3, #2) --(#1,#2) -- (#1,#4)  --  ( #3,#4);
    \node () at ($(#1 /2 + #3 /2, #2 /2 + #4 /2)$) {$\scriptstyle{#5}$};
 }

\newcommand{\ema}[1]{\ensuremath{#1}\xspace}

\newcommand{\pk}[1][k]{\ema{\beta^{(#1)}}}
\newcommand{\maxk}{\ema{K}}
\newcommand{\maxB}{\ema{B}}

\newcommand{\app}[1][k]{\ema{\text{App}^{(#1)}}}
\newcommand{\wk}[1][k]{\ema{w^{(#1)}}}
\newcommand{\io}[1][k]{\ema{\text{vol}_{\text{io}}^{(#1)}}}
\newcommand{\tio}[1][k]{\ema{\text{time}_{\text{io}}^{(#1)}}}

\newcommand{\wkreg}[1][k,i]{\ema{w^{(#1)}}}
\newcommand{\ioreg}[1][k,i]{\ema{\text{vol}_{\text{io}}^{(#1)}}}
\newcommand{\tioreg}[1][k,i]{\ema{\text{time}_{\text{io}}^{(#1)}}}

\newcommand{\tp}[1][k]{\ema{\rho^{(#1)}}}

\newcommand{\yieldy}[1][k]{\ema{\tilde{\rho}^{(#1)}}}
\newcommand{\yieldp}[1][k]{\ema{\tilde{\rho}_{\text{per}}^{(#1)}}}
\newcommand{\maxperiod}{\ema{T_{\max}}}
\newcommand{\minperiod}{\ema{T_{\min}}}
\newcommand{\bandwidth}{\ema{b}}
\newcommand{\duration}{\ema{D}}
\newcommand{\ggamma}{\ema{\gamma}}

\newcommand{\inst}[2][k]{\ema{\mathcal{I}^{(#1)}_{#2}}}

\newcommand{\init}[2]{\ema{\texttt{initW}^{(#1)}_{#2}}}
\newcommand{\initW}[2]{\init{#1}{#2}}
\newcommand{\term}[2]{\ema{\texttt{endW}^{(#1)}_{#2}}}
\newcommand{\termW}[2]{\term{#1}{#2}}
\newcommand{\initIO}[2]{\ema{\texttt{initIO}^{(#1)}_{#2}}}

\newcommand{\nk}[1][k]{\ema{n_{\text{tot}}^{(#1)}}}
\newcommand{\rk}[1][k]{\ema{r_{#1}}}
\newcommand{\dk}[1][k]{\ema{d_{#1}}}

\newcommand{\numk}[1][k]{\ema{n^{(#1)}}}

\newcommand{\band}[2]{\ema{\ggamma^{(#1)}(#2)}}
\newcommand{\bandprime}[2]{\ema{{\ggamma'}^{(#1)}(#2)}}

\newcommand{\sched}{\ema{\mathcal{P}}}
\newcommand{\period}{\ema{T}}
\newcommand{\nper}[1][k]{\ema{n_{\text{per}}^{(#1)}}}

\renewcommand{\b}{\ema{b}}

\newcommand{\intK}[2][k]{\ema{\mathcal{I}^{\{#2\}}_{#1}}}

\newcommand{\event}{\ema{\mathcal{E}}}
\newcommand{\lk}[1][k]{\ema{l_\period(#1)}}

\newcommand{\ninst}{\ema{n_{\text{inst}}}}
\newcommand{\nmax}{\ema{n_{\max}}}

\newcommand{\maxThrough}{\textsc{Sys\-Effi\-ciency}\xspace}
\newcommand{\maxThroughShort}{\textsc{SysEff}\xspace}
\newcommand{\minUserCong}{\textsc{Dilation}\xspace}

\newcommand{\periodic}{{\sc Periodic}\xspace}

\newcommand{\opt}{\ema{\text{opt}}}

\newcommand{\iis}{\textsc{Insert-In-Pattern}\xspace} 
\newcommand{\iisFirst}{\textsc{Insert-First-Instance}\xspace}
\newcommand{\persched}{\textsc{PerSched}\xspace}

\graphicspath{{fig/}}

\begin{document}
\title{Periodic I/O scheduling for super-computers}
\author{
Guillaume Aupy\thanks{Inria \& Université de Bordeaux, Talence, France} \and Ana Gainaru\thanks{Mellanox Technologies, Oak Ridge, USA} \and Valentin Le Fèvre\thanks{\'Ecole Normale Supérieure de Lyon, France}\\ 
}
\date{}
	\maketitle
	
\begin{abstract}
With the ever-growing need of data in HPC applications, the congestion at the
I/O level becomes critical in super-computers. Architectural enhancement such as
burst-buffers and pre-fetching are added to machines, but are not sufficient to
prevent congestion. Recent online I/O scheduling strategies have been put in
place, but they add an additional congestion point and overheads in the
computation of applications. 

In this work, we show how to take advantage of the periodic nature of HPC
applications in order to develop efficient periodic scheduling strategies
for their I/O transfers. Our strategy computes once during the job scheduling phase a pattern where it
defines the I/O behavior for each application, after which the applications run
independently, transferring their I/O at the specified times. Our strategy limits
the amount of I/O congestion at the I/O node level and can be easily integrated
into current job schedulers. We validate this model through extensive simulations
and experiments by comparing it to state-of-the-art online solutions, showing that
not only our scheduler has the advantage of being de-centralized and thus overcoming the
overhead of online schedulers, but also that it performs better than these
solutions, improving the application dilation up to 13\% and the maximum
system efficiency up to 18\%.
\end{abstract}
	
\section{Introduction}
	\label{sec.intro}

In the race to larger supercomputers, the most commonly used metric is is the
computational power. However supercomputers are not simply computers with
billions of processors. One of the reason why Sunway TaihuLight (the world
fastest supercomputer as of Nov 2016~\cite{top500}), reaches 93~PetaFlops on HPL (a performance benchmark
based on dense linear algebra), but struggles to reach 0.37~PetaFlop on HPCG, a
recent benchmark based on actual HPC applications~\cite{dongarra2013toward} is
data movement. 
Nowadays, a supercomputing application creates or has to deal with TeraBytes of
data. This is true in all fields, from medical research (Brain initiatives), to
astrophysics (HACC~\cite{habib2012universe}, Enzo~\cite{enzoenzo},
HOMME~\cite{nair2007petascale}), including meteorology
(CM1~\cite{bryan2002benchmark}) and fusion plasma
(GTC~\cite{ethier2012petascale}).
In 2013, Argonne upgraded its house supercomputer: moving from Intrepid (peak
performance: 0.56~PFlops; peak I/O throughput: 88~GB/s) to Mira (peak
performance: 10~PFlops; peak I/O throughput: 240~GB/s). While both criteria
seem to have improved considerably, the reality behind is that for a given
application, its I/O throughput scales linearly (or worse) with its performance,
and hence, what should be noticed is a downgrade from 160~GB/PFlop to
24~GB/PFlop! On Intrepid, it was shown that I/O congestion could cause up to a
70\% decrease to the I/O throughput~\cite{gainaru2015scheduling}.

To help with the ever growing amount of data created, architectural improvement
such as burst buffers~\cite{Liu12onthe} have been added to the system. Work is being
done to transform the data before sending it to the disks in the hope of
reducing the I/O sent~\cite{matthieu}. However, even with the current I/O footprint burst buffers are not able to completely hide congestion. Moreover, the data used is always expected
to grow. Recent works~\cite{gainaru2015scheduling} have started working on novel
online, centralized I/O scheduling strategies at the I/O node level. However
one of the risk noted on these strategies is the scalability issue caused by
potentially high overheads (between 1 and 5\% depending on the number of
nodes used in the experiments)~\cite{gainaru2015scheduling}. Moreover, it is expected this overhead to increase at larger scale since it need centralized information about all applications running in the system.

In this paper, we present a decentralized I/O scheduling strategy for
super-computers. We show how to take known HPC application behaviors (namely
their periodicity) into account to derive novel static algorithms. The
periodicity of HPC applications has been well observed and
documented~\cite{carns200924,gainaru2015scheduling,dorier2014omnisc}: HPC
applications alternate between computation and I/O transfer, this pattern
being repeated over-time. Furthermore, fault-tolerance technique (such as periodic
checkpointing~\cite{daly04}) also add to this periodic behavior. 
Using this periodicity property, we compute a static periodic scheduling
strategy, which provides a way for each applications to know when they should
start transferring their I/O (i) hence reducing potential bottlenecks either due
to I/O congestion, and (ii) without having to consult with I/O nodes every time
I/O should be done and hence adding an extra overhead. The main contributions of
this paper are:
\begin{compactitem}
	\item A novel light-weight I/O algorithm that looks at optimizing both
application-oriented (dilation or fairness) and platform-oriented (maximum
system efficiency) objectives;
	\item A set of extensive simulations and experiments that show that this
algorithm performs as well or better than current state of the art heavy-weight
online algorithms.
\end{compactitem}
Note that the algorithm presented here is done as a proof of concept to show the
efficiency of this kind of light-weight techniques. We believe our scheduler can
be implemented naturally into a job scheduler and we provide experimental
results backing this claim. However, this integration is beyond the scope of
this paper.

The rest of the paper is organized as follows: in Section~\ref{sec.model} we
present the application model and optimization problem. In
Section~\ref{sec.algo} we present our novel algorithm technique as well as
a brief proof of concept for a future implementation. In Section~\ref{sec.simu}
we present extensive simulations based on the model to show the performance of
our algorithm compared to state of the art. We then confirm the performance on a
super-computer to validate the model. We give some background and related work
in Section~\ref{sec.related}. We provide concluding remarks and ideas for future
research directions in Section~\ref{sec.conclusion}.

	\section{Model}
	\label{sec.model}

In this section we use the model introduced in our previous
work~\cite{gainaru2015scheduling} that has been verified experimentally to be
consistent with the behavior of Intrepid and Mira, super-computers at Argonne.

We consider scientific applications running at the same time on a parallel
platform. The applications consist of series of computations followed by I/O operations. On a super-computer, the computations are done
independently because each application uses its own nodes. However, the
applications are concurrently sending and receiving data during their I/O phase
on a dedicated I/O network. The consequence of this I/O concurrency is
congestion between an I/O node of the platform and the file storage.

\subsection{Parameters}
\label{sec.model.param}

We assume that we have a parallel platform made up of $N$ identical unit-speed
nodes, composed of the same number of identical processors, each equipped with
an I/O card of bandwidth \bandwidth (expressed in bytes per second). We further
assume a centralized I/O system with a total bandwidth \maxB (also expressed in
bytes per second). This means that the total bandwidth between the computation
nodes and an I/O node is $N \cdot \bandwidth$ while the bandwidth between an I/O
node and the file storage is \maxB, with usually $N\cdot \bandwidth \gg \maxB$.
We have instantiated this model for the Intrepid platform on
Figure~\ref{fig:intrepid}.

\begin{figure}[tbh]
\centering 
\includegraphics[width=0.8\columnwidth]{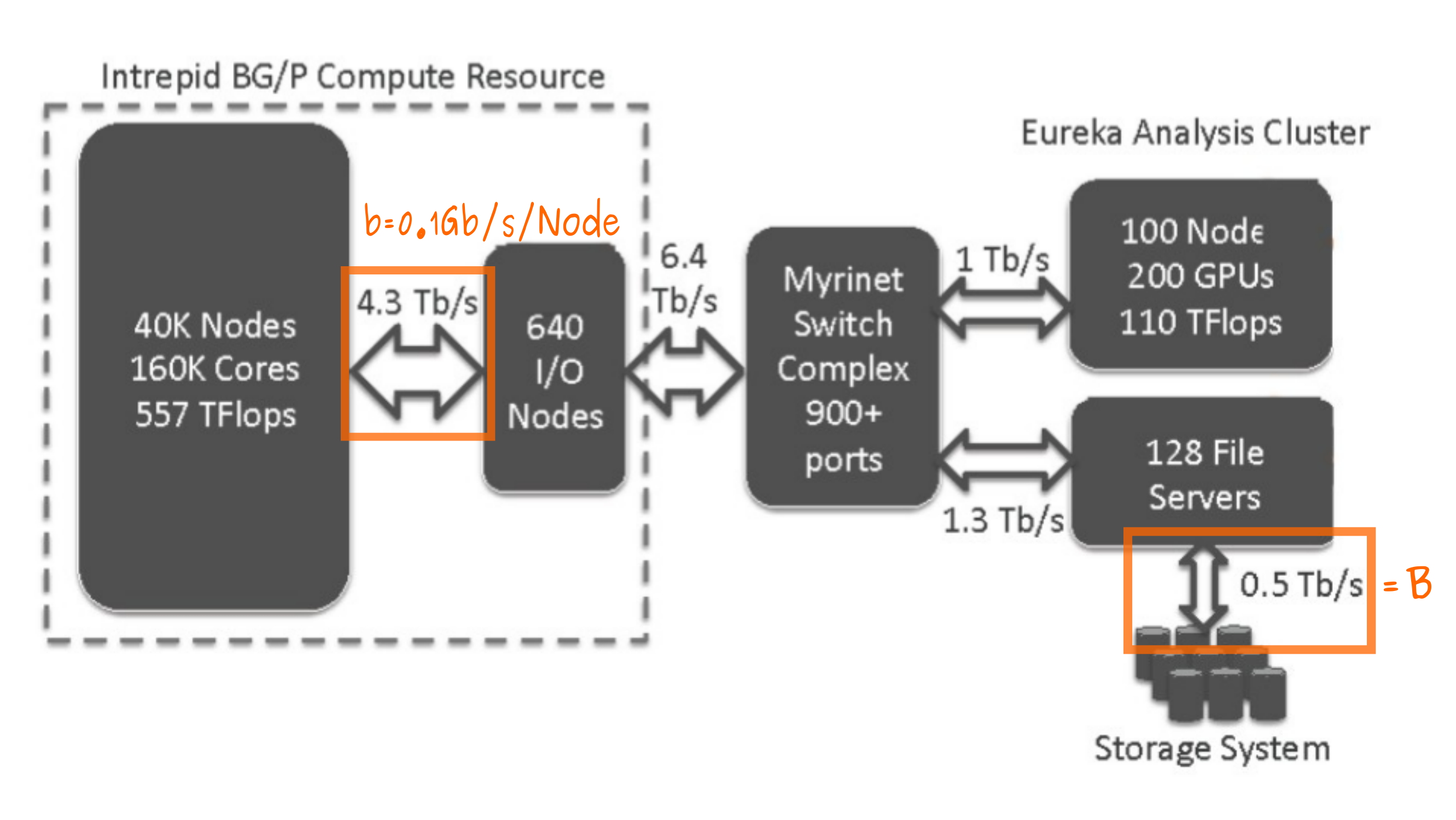}
\caption{Model instantiation for the Intrepid
platform~\cite{gainaru2015scheduling}.}
\label{fig:intrepid}
\end{figure}

We have \maxk applications, all assigned to independent and dedicated
computational resources, but competing for I/O. For each application \app we
define:
\begin{compactitem}
	\item Its size: \app executes with \pk dedicated processors;
	\item Its pattern: \app obeys a pattern that repeats over time. There are
\nk \emph{instances} of \app that are executed one after the other. Each
instance consists of two disjoint phases: computations that takes a time \wk,
followed by I/O transfers for a total volume \io. The next instance cannot start
before I/O operations for the current instance is terminated.
\end{compactitem}
We further denote by \rk the time when \app is released on the platform and \dk
the time when the last instance is completed.
Finally, we denote by \band{k}{t}, the bandwidth used by a node on which
application \app is running, at instant $t$.


\subsection{Execution Model}
\label{sec.model.exec}

As the computation resources are dedicated, we can always assume w.l.o.g that
the next computation chunk starts right away after completion of the previous
I/O transfers, and is executed at full (unit) speed. On the contrary, all
applications compete for I/O, and congestion will likely occur. 
The simplest case is that of a single periodic application \app using the I/O
system in dedicated mode during a time-interval of duration $\duration$. In that
case, let \ggamma be the I/O bandwidth used by each processor of \app during
that time-interval. We derive the condition 
$\pk \ggamma \duration = \io$
to express that the entire I/O data volume is transferred. We must also enforce the constraints that 
(i) $\ggamma \leq \bandwidth$ (output capacity of each processor); and (ii)  
$\pk \ggamma \leq \maxB$ (total capacity of I/O system). Therefore, the minimum time to perform 
the I/O transfers for an instance of \app is 
$\tio = \frac{\io}{\min(\pk \bandwidth, \maxB)}$.
However, in general many applications will use the I/O system simultaneously,
whose bandwidth capacity \maxB will be shared among all these applications (see
Figure~\ref{fig:model}).

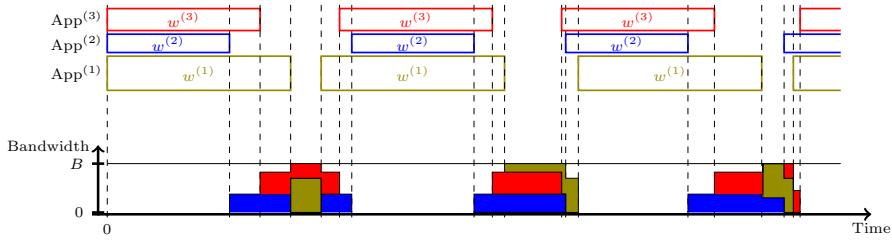
\begin{figure}[tbh]
\newcommand{\dashedline}[1]{\draw[dashed] (#1,70) -- (#1,-108);}
\resizebox{\columnwidth}{!}{
\begin{tikzpicture}
\begin{scope}[scale=1]
\begin{scope}[xscale = 1/10, yscale=1/50]

\dashedline{0}

\node () at (-5,14) {$\scriptstyle{\app[1]}$};
\draw[-, thick,olive] (0,0) rectangle (30,28) node [pos=.5] {$\scriptstyle{\wk[1]}$};
\dashedline{30}
\dashedline{35}
\draw[-, thick,olive] (35,0) rectangle (65,28) node [pos=.5] {$\scriptstyle{\wk[1]}$};
\dashedline{65}
\dashedline{77+1/14}
\draw[-, thick,olive] (77+1/14,0) rectangle (107+1/14,28) node [pos=.5] {$\scriptstyle{\wk[1]}$};
\dashedline{107+1/14}
\dashedline{112+1/4}
\draw[-, thick,olive] (120,0) -- (112+1/4,0) -- (112+1/4,28) -- (120,28) ;

\node () at (-5,38) {$\scriptstyle{\app[2]}$};
\draw[-, thick,blue] (0,31) rectangle (20,46) node [pos=.5] {$\scriptstyle{\wk[2]}$};
\dashedline{20}
\dashedline{40}
\draw[-, thick,blue] (40,31) rectangle (60,46) node [pos=.5] {$\scriptstyle{\wk[2]}$};
\dashedline{60}
\dashedline{75}
\draw[-, thick,blue] (75,31) rectangle (95,46) node [pos=.5] {$\scriptstyle{\wk[2]}$};
\dashedline{95}
\dashedline{110 +41/56}
\draw[-, thick,blue] (120,31) -- (110+41/56,31) -- (110+41/56,46) -- (120,46) ;

\node () at (-5,58) {$\scriptstyle{\app[3]}$};
\draw[-, thick,red] (0,49) rectangle (25,67) node [pos=.5] {$\scriptstyle{\wk[3]}$};
\dashedline{25}
\dashedline{38}
\draw[-, thick,red] (38,49) rectangle (63,67) node [pos=.5] {$\scriptstyle{\wk[3]}$};
\dashedline{63}
\dashedline{74+1/3}
\draw[-, thick,red] (74+1/3,49) rectangle (99+1/3,67) node [pos=.5] {$\scriptstyle{\wk[3]}$};
\dashedline{99+1/3}
\dashedline{113.35}
\draw[-, thick,red] (120,49) -- (113.35,49) -- (113.35,67) -- (120,67) ;

\begin{scope}[yshift=-100cm]
\draw[->, very thick] (\ttrdy,-2.5) -- (\ttrdy,55) node [pos=1, left] {$\scriptstyle{\text{Bandwidth}}$};
\draw[->, very thick] (-2.5,\ttrdy) -- (125,\ttrdy) node [pos=1, below] {$\scriptstyle{\text{Time}}$};

\xlegend{0}{0};
\draw[-, very thick] ($(\ttrdy-1, 0)$) -- ($(\ttrdy +1, 0)$) node [pos=0, left] {$\scriptstyle{0}$};
\draw[-, very thick] ($(\ttrdy -1, 40)$) -- ($(\ttrdy +1, 40)$) node [pos=0, left] {$\scriptstyle{\maxB}$};
\draw[ultra thin] (0,40) -- (120,40);

\draw[-, ultra thin,fill=olive] (30,0) rectangle (35,28) ;
\draw[-, ultra thin,fill=olive] (65,33) -- (65,40) -- (75,40) -- (75,28) -- ($(75+58/28, 28)$) -- ($(75+58/28, 0)$) -- (75, 0) -- (75,15) -- (223/3, 15) -- (223/3, 33);
\draw[-, ultra thin,fill=olive] (107+1/4,12) -- (107+1/4,40) -- (110+41/56,40) --  (110+41/56,28) --  (112+1/4, 28) -- (112+1/4,0)-- (110+41/56,0) -- (110+41/56,12) ;

\draw[-, ultra thin,fill=blue] (20,0) -- (20,15) -- (30,15) -- (30,0) -- (35,0) -- (35,15) -- (40,15) -- (40,0) ;
\draw[-, ultra thin,fill=blue] (60,0) rectangle (75,15);
\draw[-, ultra thin,fill=blue] (95,0) -- (95,15) -- (107+1/14,15) -- (107+1/14,12) -- (110+41/56,12) -- (110 +41/56,0) ;

\draw[-, ultra thin,fill=red] (25,15) -- (25,33) -- (30,33) -- (30,40) -- (35,40) -- (35,33) -- (38,33) -- (38,15) -- (35,15) -- (35,28) --  (30,28) -- (30,15);
\draw[-, ultra thin,fill=red] (63,15) rectangle ($(223/3,33)$);
\draw[-, ultra thin,fill=red] (99+1/3,15) -- (99+1/3,33) -- (107+1/4,33) -- (107+1/4,40) -- (110+41/56,40) -- (112+1/4,40) -- (112+1/4,18) -- (113.35,18) -- (113.35,0) -- (112+1/4,0) -- (112+1/4,28) -- (110+41/56,28) -- (110+41/56,40) -- (107+1/4,40) -- (107+1/4,15) ;
\end{scope}

\end{scope}
\end{scope}
\end{tikzpicture}
}
\caption{Scheduling the I/O of three periodic applications (top: computation, bottom: I/O).}
\label{fig:model}
\end{figure}

This model is very flexible, and the only assumption is that at any instant, all
processors assigned to a given application are assigned the same bandwidth. This
assumption is transparent for the I/O system and simplifies the problem
statement without being restrictive.
Again, in the end, the total volume of I/O transfers for an instance of \app
must be \io, and at any instant, the rules of the game
are simple: never exceed the individual bandwidth \bandwidth of each processor ($\band{k}{t} \leq \bandwidth$ for any $k$ and $t$), 
and never exceed the total bandwidth \maxB of the I/O system ($\sum\limits_{k=1}^\maxk \pk \band{k}{t} \leq \maxB$ for any $t$).

	\subsection{Objectives}
We now focus on the optimization objectives at hand here. We use the objectives
introduced in~\cite{gainaru2015scheduling}. 

First, the \emph{application efficiency} achieved for each application \app at
time $t$ is defined as
\[\yieldy(t)  = \frac{\sum_{i \leq \numk\!(t)} \wkreg}{t - \rk},\]
where $\numk\!(t) \leq \nk$ is the number of instances of application \app that have been executed at time $t$,
since the release of \app at time \rk. Because we execute \wkreg units of computation 
followed by \ioreg units of I/O operations on instance \inst{i} of \app, we have  
$t - \rk \geq \sum_{i \leq  \numk(t)} \left (\wkreg + \tioreg \right )$.
Due to I/O congestion, \yieldy never exceeds the optimal efficiency that can be
achieved for \app, namely
\[\tp = \frac{\wk}{\wk + \tio}\]

The two key optimization objectives, together with a rationale for each of them,
are:
\begin{compactitem}
	\item \maxThrough: where we maximize the peak performance of the platform,
namely maximizing the amount of operations per time unit:
\begin{equation}
\label{eq:syseff}
\text{maximize } \frac{1}{N} \sum_{k=1}^\maxk \pk \yieldy(\dk) .
\end{equation}
	\item \minUserCong: where we minimize the largest slowdown imposed to each
application (hence optimizing fairness across applications):
\begin{equation}
\label{eq:dilation}
\text{minimize } \max_{k=1 .. \maxk} \frac{\tp}{\yieldy(\dk)} .
\end{equation}
\end{compactitem}


Note that it is known that both problems are NP-complete, even in an (easier)
offline setting~\cite{gainaru2015scheduling}.

	\section{Periodic scheduling strategy}
	\label{sec.algo}

In general, for an application \app, \nk the number of instances of \app is very
large and not polynomial in the size of the problem. For this reason, online
schedule have been preferred until now.		
The key novelty of this paper is to introduce {\em periodic schedules} for the
\maxk applications. Intuitively, we are looking for a computation and I/O {\em
pattern} of duration \period that will be repeated over time (except for {\em
initialization} and {\em clean up} phases), as shown on 
Figure~\ref{fig.period.general}. 
In this section, we start by introducing the notion of periodic schedules and a
way to compute the application efficiency differently. We then provide the
algorithms that are at the core of this work.
	
Because there is no competition on computation (no shared resources),
we can consider that a chunk of computation directly follows the end of the
transfer of I/O, hence we need only to represent I/O transfers in this pattern.
The bandwidth used by each application during the I/O operations is represented
over time, as shown in Figure~\ref{fig.period.detail}.
We can see that an operation can overlap with the one of the previous
pattern or the next pattern, but overall, the pattern will just repeat.
	
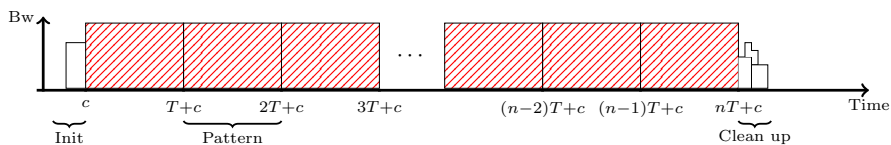
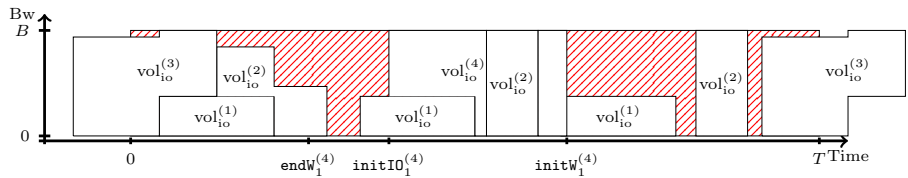
\begin{figure}[htb]
\centering
\subfloat[Periodic schedule (phases)]{\resizebox{\linewidth}{!}{
\begin{tikzpicture}
\begin{scope}[scale=1]
\begin{scope}[xscale = 1/10, yscale=1/50]

\draw[->, very thick] (\ttrdy,-2.5) -- (\ttrdy,55) node [pos=1, left] {$\scriptstyle{\text{Bw}}$};
\draw[->, very thick] (-2.5,\ttrdy) -- (125,\ttrdy) node [pos=1, below] {$\scriptstyle{\text{Time}}$};

\RRect{2}{0}{5}{35}{};
\draw [thick, decoration={brace, mirror,raise=0.5cm}, decorate] (0,0) -- (5,0); 
    \node () at ($(2.5,25*\ttrdy)$) {$\scriptstyle{\text{Init}}$};
\draw[-, ultra thin,pattern=north east lines, pattern color=red] (5,0) rectangle (20,50) ;
\draw[-, ultra thin,pattern=north east lines, pattern color=red] (20,0) rectangle (35,50) ;
\draw[-, ultra thin,pattern=north east lines, pattern color=red] (35,0) rectangle (50,50) ;
    \node () at (55,25) {$\cdots$};
\draw[-, ultra thin,pattern=north east lines, pattern color=red] (60,0) rectangle (75,50) ;
\draw[-, ultra thin,pattern=north east lines, pattern color=red] (75,0) rectangle (90,50) ;
\draw[-, ultra thin,pattern=north east lines, pattern color=red] (90,0) rectangle (105,50) ;
\draw [thick, decoration={brace, mirror,raise=0.5cm}, decorate] (20,0) -- (35,0); 
    \node () at ($(27.5,25*\ttrdy)$) {$\scriptstyle{\text{Pattern}}$};

\LRect{105}{0}{107}{24}{};
\draw[fill=white, draw = black] (106,24) -- (106,35) -- (107,35) -- (107,29) -- (108,29) -- (108,18) -- (107,18) -- (107, 24) ; 
\draw[fill=white, draw = black] (107,0) rectangle (109.5,18) node [pos=.5] {} ; 
\draw [thick, decoration={brace, mirror,raise=0.5cm}, decorate] (105,0) -- (110,0); 
    \node () at ($(107.5,25*\ttrdy)$) {$\scriptstyle{\text{Clean up}}$};

\xlegend{5}{c};
\xlegend{20}{\period+c};
\xlegend{35}{2 \period+c};
\xlegend{50}{3\period+c};
\xlegend{75}{(n-2) \period+c};
\xlegend{90}{(n-1)\period+c};
\xlegend{105}{n\period+c};
\end{scope}
\end{scope}
\end{tikzpicture}
}
\label{fig.period.general}}

\subfloat[Detail of I/O in a period/pattern]{
\resizebox{\linewidth}{!}{
\begin{tikzpicture}
\begin{scope}[scale=1]
\begin{scope}[xscale = 1/11.5, yscale=1/20]

\draw[->, very thick] (\ttrdx,-2.5) -- (\ttrdx,37) node [pos=1, left] {$\scriptstyle{\text{Bw}}$};
\draw[->, very thick] ($(\ttrdx-1,\ttrdy)$) -- (125,\ttrdy) node [pos=1, below] {$\scriptstyle{\text{Time}}$};
\draw[-, ultra thin,pattern=north east lines, pattern color=red] (0,0) rectangle (120,32) ;
\xlegend{0}{0};
\ylegend{0}{0};
\xlegend{120}{\period};
\ylegend{32}{\maxB};

\draw[fill=white, draw = black] (5,0) rectangle (25,12) node [pos=.5] {$\scriptstyle{\io[1]}$} ;
\begin{scope}[xshift=35cm]
\draw[fill=white, draw = black] (5,0) rectangle (25,12) node [pos=.5] {$\scriptstyle{\io[1]}$} ; 
\end{scope}
\begin{scope}[xshift=70cm]
\draw[fill=white, draw = black] (5,0) rectangle (25,12) node [pos=.5] {$\scriptstyle{\io[1]}$} ; 
\end{scope}

\draw [black,fill=white]  (15,12) -- ( 15, 27) -- (25,27) -- (25, 15) -- ( 34.2,15) -- (34.2,0) -- (25,0) -- (25,12);
\node () at (20,18) {$\scriptstyle{\io[2]}$};
\draw[fill=white, draw = black] (62,0) rectangle (71,32) node [pos=.5] {$\scriptstyle{\io[2]}$} ; 
\draw[fill=white, draw = black] (98.5,0) rectangle (107.5,32) node [pos=.5] {$\scriptstyle{\io[2]}$} ; 

\draw [black,fill=white]  (-10,0) -- ( -10, 30) -- (5,30) -- (5, 32) -- ( 15,32) -- (15,12) -- (5,12)-- (5,0) -- (-10, 0);
\node () at (5,20) {$\scriptstyle{\io[3]}$};
\begin{scope}[xshift=120cm]
\draw [black,fill=white]  (-10,0) -- ( -10, 30) -- (5,30) -- (5, 32) -- ( 15,32) -- (15,12) -- (5,12)-- (5,0) -- (-10, 0);
\node () at (5,20) {$\scriptstyle{\io[3]}$};
\end{scope}

\draw [black,fill=white]  (45,12) -- ( 45, 32) -- (62,32) -- (62, 0) -- ( 60,0) -- (60,12);
\node () at (58,20) {$\scriptstyle{\io[4]}$};
\draw [black,fill=white]  (71,0) -- ( 71, 32) -- (76,32) -- (76, 0) -- (71,0);

\xlegend{76}{\initW{4}{1}};
\xlegend{31}{\termW{4}{1}};
\xlegend{45}{\initIO{4}{1}};

\end{scope}
\end{scope}
\end{tikzpicture}}
\label{fig.period.detail}}
\caption{A schedule (above), and the detail of one of its regular pattern
(below), where
$(\wk[1] = 3.5; \io[1]=240; \nper[1]=3) $, 
$(\wk[2] = 27.5; \io[2]=288;  \nper[2]=3) $, 
$(\wk[3] = 90; \io[3]=350; \nper[3]=1) $, 
$(\wk[4] = 75; \io[4]=524; \nper[4]=1) $.}
\label{fig.period}
\end{figure}

To describe a pattern, we use the following notations:
\begin{compactitem}
	\item \nper: the number of instances of \app during a pattern.
	\item \inst{i}: the $i$-th instance of \app during a pattern.
	\item $\init{k}{i}$: the time of the beginning of \inst{i}. So, \inst{i} has a computation interval going from $\init{k}{i}$ to $\term{k}{i} = \init{k}{i} +\wk\text{ mod } T$.
	\item \initIO{k}{i}: the time when the I/O transfer from the $i$-th instance of \app starts (between \term{k}{i} and \initIO{k}{i}, \app is idle). Therefore, we have
	\[\int_{\initIO{k}{i}}^{\init{k}{(i+1)\%\nper}} \pk \band{k}{t}dt = \io.\]
\end{compactitem}
Globally, if we consider the two dates per instance \init{k}{i} and
\initIO{k}{i}, that define the change between computation and I/O phases, we
have a total of $S \leq \sum_{k=1}^\maxk 2\nper$ distinct dates, that are called
the \emph{events} of the pattern.

We define the periodic efficiency of a pattern of size \period:
\begin{equation}
\yieldp = \frac{\nper \wk}{\period}.
\end{equation}
For periodic schedules, we use it to approximate the actual efficiency
achieved for each application.
The rationale behind this can be seen on Figure~\ref{fig.period}.
If \app is released at time \rk, and the first pattern starts at time $\rk+c$, that is after an initialization
phase, then the main pattern is repeated $n$ times (until time $n\cdot\period + \rk + c$), and 
finally \app ends its execution after a clean-up phase at time $\dk=\rk + c + n\cdot\period + c'$.
If we assume that $n\cdot\period \gg c+c'$, then $\dk - \rk \approx n\cdot\period$.
Then the value of the $\yieldy(d_k)$ for \app is: 
\begin{align*}
\yieldy(d_k) &=\frac{\left(n \cdot \nper + \delta \right )\wk}{d_k - r_k} = \frac{\left(n \cdot \nper + \delta \right )\wk}{c + n\cdot\period + c'} \\
& \approx \frac{\nper \wk}{\period} = \yieldp
\end{align*}
where $\delta$ can be 1 or 0 depending whether \app was executed or not during the clean-up or 
init phase.

	\subsection{\persched: a periodic scheduling algorithm}
	\label{sec.algo.persched}
{\em For details in the implementation, we refer the interested reader to the
source code available at \url{https://github.com/vlefevre/IO-scheduling-simu}.}

The difficulties of finding an efficient periodic schedule are three-fold:
\begin{compactitem}
	\item The first one is that the right pattern size has to be determined;
	\item The second one is that for a given pattern size, the number of instances of
each application that should be included in this pattern need to be determined;
	\item Finally, the time constraint between two consecutive I/O transfers of a
given application, due to the computation in-between makes naive scheduling
strategies harder to implement.
\end{compactitem}

\paragraph{Finding the right pattern size}
A solution is to find schedules with different pattern sizes between a minimum
pattern size \minperiod and a maximum pattern size \maxperiod.

Because we want a pattern to have at least one instance of each application, we
can trivially set up $\minperiod = \max_k (\wk +\tio)$.
Intuitively, the larger \maxperiod is, the more possibilities we can have to
find a good solution. However this also increases the complexity of the
algorithm. We want to limit the number of instances of all applications in a
schedule. For this reason we chose to have $\maxperiod=O(\max_k (\wk +\tio))$.
We discuss this hypothesis in Section~\ref{sec.simu}, where we give better
experimental intuition on finding the right value for \maxperiod. Experimentally
we observe (Section~\ref{sec.simu}, Figure~\ref{fig.tmax}) that
$\maxperiod=10\minperiod$ seems to be sufficient.

We then decided on an iterative search where the pattern size increases
exponentially at each iteration from \minperiod to \maxperiod.
In particular, we use a precision $\varepsilon$ as input and we iteratively
increase the pattern size from \minperiod to \maxperiod by a factor $(1+\varepsilon)$.
This allows us to have a polynomial number of iterations.
The rationale behind the exponential increase is that when the pattern size gets
large, we expect performance to converge to an optimal value, hence needing less
the precision of a precise pattern size. Furthermore while we could try only large
pattern size, it seems important to find a good small pattern size as it would simplify
the scheduling step. Hence a more precise search for smaller pattern sizes. Finally,
we expect the best performance to cycle with the pattern size. We verify these
statements experimentally in Section~\ref{sec.simu} (Figure~\ref{fig.evolutionT}).

\paragraph{Determining the number of instances of each application}
By choosing $\maxperiod=O(\max_k (\wk +\tio))$, we guarantee the maximum number
of instances of each application that fit into a pattern is 
$O\left (\frac{\max_k (\wk +\tio)}{\min_k (\wk +\tio)}\right )$.

\paragraph{Instance scheduling}
Finally, our last item is, given a pattern of size \period, how to schedule
instances of applications into a periodic schedule.

To do this, we decided on a strategy where we insert instances of applications
in a pattern, without modifying dates and bandwidth of already scheduled
instances. Formally, we call an application schedulable:
\begin{definition}[Schedulable]	
\label{def:schedulable}
  Given an existing pattern 
  \[\sched = \cup_{k=1}^{\maxk} \left ( \nper, \cup_{i=1}^{\nper} \{\init{k}{i},\initIO{k}{i},\band{k}{}\}\right ),\] 
  we say that an application \app is schedulable if there exists $1 \leq i \leq
  \nper$, such that: 
\begin{equation}
	\label{eq:schedulable}
\int_{\init{k}{i}+\wk}^{\initIO{k}{i}-\wk} \min\left (\pk \bandwidth,\maxB - \sum_l \pk[l] \band{l}{t} \right )dt \geq \io
\end{equation}
\end{definition}

\begin{figure}[tbh]
\centering  
\resizebox{0.6\textwidth}{!}{

\begin{tikzpicture}
\begin{scope}[scale=1]
\begin{scope}[xscale = 1/11.5, yscale=1/20]

\clip (7.5,-15) rectangle (85,40);
\draw[->, very thick] (13,-2.5) -- (13,37) node [pos=1, left] {$\scriptstyle{\text{Bw}}$};
\draw[->, very thick] ($(\ttrdx-1,\ttrdy)$) -- (80,\ttrdy) node [pos=1, below] {$\scriptstyle{\text{Time}}$};
\draw[-, ultra thin] (0,0) rectangle (90,32) ;
\draw (10,32) -- (80,32);
\xlegend{0}{0};
\ylegend{0}{0};
\ylegend{32}{\maxB};

\draw[fill=white, draw = black] (5,0) rectangle (25,12) node [pos=.5] {$\scriptstyle{\io[1]}$} ;
\begin{scope}[xshift=35cm]
\draw[fill=white, draw = black] (5,0) rectangle (25,12) node [pos=.5] {$\scriptstyle{\io[1]}$} ; 
\end{scope}

\draw [black,fill=white]  (15,12) -- ( 15, 27) -- (25,27) -- (25, 15) -- ( 34.2,15) -- (34.2,0) -- (25,0) -- (25,12);
\node () at (20,18) {$\scriptstyle{\io[2]}$};
\draw[fill=white, draw = black] (62,0) rectangle (71,32) node [pos=.5] {$\scriptstyle{\io[2]}$} ; 
\draw[fill=white, draw = black] (98.5,0) rectangle (107.5,32) node [pos=.5] {$\scriptstyle{\io[2]}$} ; 

\begin{scope}[xshift=120cm]
\draw [black,fill=white]  (-10,0) -- ( -10, 30) -- (5,30) -- (5, 32) -- ( 15,32) -- (15,12) -- (5,12)-- (5,0) -- (-10, 0);
\node () at (5,20) {$\scriptstyle{\io[3]}$};
\end{scope}

\draw[pattern color=blue, pattern=north east lines,draw=black] (37.2,0) -- (37.2,32) -- (59,32) -- (59,12) -- (40,12) -- (40,0) -- (37.2,0);
\draw[pattern color=red, pattern=north west lines,draw=black] (34.2,0) rectangle (37.2,32);
\draw[pattern color=red, pattern=north west lines,draw=black] (59,12) -- (59,32) -- (62,32) -- (62,0) -- (60,0) -- (60,12) -- (59,12);

\xlegend{37.2}{\initW{2}{2}+\wk[2]};
\xlegend{59}{\initIO{2}{2}-\wk[2]};

\end{scope}
\end{scope}
\end{tikzpicture}
}
\caption{Description of what schedulable means: if we want to insert an instance of \app[2], we need to check that the blue area is greater than \io[2], while the red area is reserved
for the computation of \wk[2].}
\label{fig.schedulable-def}
\end{figure}
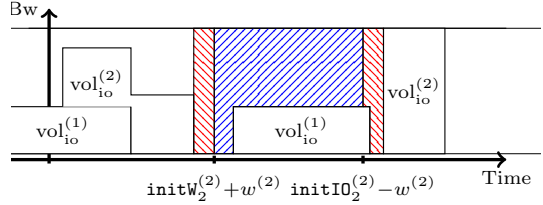
To understand the integral in Equation~\eqref{eq:schedulable}: we are checking
that during the end of the computation of the $i^{\text{th}}$ instance
($\init{k}{i}+\wk$), and the beginning of the computation of the
$i+1^{\text{th}}$ instance (\initIO{k}{i}-\wk), there is enough bandwidth to
perform at least a volume of I/O of \io. We represent it graphically on
Figure~\ref{fig.schedulable-def}.


\begin{algorithm}[tbh]
\small
\caption{\iis}
\label{algo:iis}
procedure {\iis}(\sched, \app)\\
\Begin{
\eIf{\app has 0 instance}{
	\Return{\iisFirst(\sched,\app)}\;
}
{
	$T_{\min}:=+\infty$ \;
	Let $\intK{i}$ be the last inserted instance of \app\;
	Let $\event_0, \event_1, \cdots, \event_{j_i}$ the times of the events between the end of $\intK{i}+\wk$ and the beginning of $\intK{(i+1)\text{ mod }\lk}$\;
	For $l = 0 \cdots j_i-1$, let $B_l$ be the minimum between \pk\b and the available bandwidth during $[\event_l, \event_{l+1}]$\;
	DataLeft = \io\;
	$l = 0$\;
	sol = []\;
	\While{DataLeft $> 0$ and $l < j_i$}{
		\label{iis.line.while}
		\If{$B_l > 0$}{
			TimeAdded $= \min(\event_{l+1}-\event_l, \text{DataLeft}/B_{l})$\;
			DataLeft -= TimeAdded$\cdot B_{l}$\;
			sol = $[ (\event_l, \event_l+TimeAdded, B_l) ]$ + sol\;
		}
		$l$++\;
	}
	\eIf{DataLeft$>0$}{\Return{\sched}}{\Return{\sched.\texttt{addInstance}(\app,sol)}}
}
}
\end{algorithm}

With Definition~\ref{def:schedulable}, we can now explain the core idea of the
instance scheduling part of our algorithm.
Starting from an existing pattern, while there exist applications that are
schedulable:
\begin{compactitem}
	\item Amongst the applications that are schedulable, we choose the
application that has the worse \minUserCong. The rationale is that even though we want to
increase \maxThrough, we do it in a way that ensures that all applications are
treated fairly;
	\item We insert the instance into an existing scheduling using a procedure
\iis such that (i) the first instance of each application is inserted using
procedure \iisFirst which minimizes the time of the I/O transfer of this new
instance, (ii) the other instances are inserted just after the last inserted
one. 
\end{compactitem}

Note that \iisFirst is implemented using a water-filling
algorithm~\cite{gallager1968information} and \iis is implemented as described in
Algorithm~\ref{algo:iis}.
We use a different function for the first instance of each application because
we do not have any previous instance to use the \iis function. Thus, the basic
idea would be to put them at the beginning of the pattern, but it will be more
likely to create congestion if all applications are ``synchronized'' (for
example if all the applications are the same, they will all start their I/O
phase at the same time). By using \iisFirst, every first instance will be at a
place where the congestion for it is minimized. This creates a starting point
for the subsequent instances.

The function \texttt{addInstance} updates the pattern with the new instance,
given a list of the intervals $(\event_l, \event_{l'}, \bandwidth_l)$ during
which \app transfers I/O between $\event_l$ and $\event_{l'}$ using a bandwidth
$\bandwidth_l$.

\paragraph{Correcting the period size}
In Algorithm~\ref{algo:persched}, the pattern sizes under trial are determined by
\minperiod and $\varepsilon$. There is no reason why this would be the right
pattern size, and one might be interested in reducing it to fit precisely the
instances that are included in the solutions that we found.

In order to do so, once a periodic pattern has been computed, we try to improve
the best pattern size we found in the first loop of the algorithm, by trying new
pattern sizes, close to the previous best one, let us say $T_{\opt}$. To do
this, we add a second loop which now tries $1/\varepsilon$ uniformly distributed
pattern sizes from $T_{\opt}$ to $T_{\opt}/(1+\varepsilon)$.

With all of this in mind, we can now write \persched
(Algorithm~\ref{algo:persched}), our algorithm to construct a periodic pattern.
For all pattern sizes tried between \minperiod and \maxperiod, we return the pattern
with maximal \maxThrough.

\begin{algorithm}[tbh]
\small
\caption{\small Periodic Scheduling heuristic: \persched}
\label{algo:persched}
\algorithmfootnote{We estimate \maxThrough of a periodic pattern, by replacing $\yieldy(\dk)$ by \yieldp in Equation~\eqref{eq:syseff}}
procedure {\persched}$(K',\varepsilon,\{\app\}_{1\leq k\leq \maxk})$\\
\Begin{
	$\minperiod \gets \max_k (\wk +\tio)$\;
	$\maxperiod \gets K'\cdot\minperiod$\;
	$T = \minperiod$\;
	$\texttt{SE}\gets 0$\;
	$T_{\opt} \gets 0$\;
	$\sched_{\opt} \gets \{\}$\;
	\While{$T \leq \maxperiod$}
	{\label{algo:while_main}
		\sched = \{\}\;
		\While{exists a schedulable application}
		{\label{algo:persched_while}
			$\mathcal{A} = \{\app|\app \text{ is schedulable}\}$\;
			Let \app be the element of $\mathcal{A}$ minimal with respect to
			the lexicographic order $\left(\frac{\tp}{\yieldp},\frac{\wk}{\tio}\right)$\label{algo:persched_min}\;
			$\sched\!\gets$\iis\!\!$(\sched,\app)$\;
		}
		\If{$\texttt{SE}<\maxThrough(\sched)$}{\label{algo:se}
			$\texttt{SE}\gets \maxThrough(\sched)$\;
			$T_{\opt} \gets T$\;
			$\sched_{\opt} \gets \sched$
		}
		$T \gets T \cdot (1+\varepsilon)$\;
	}
	$T \gets T_{\opt}$\;
	\While{\texttt{true}}
	{\label{algo:while_main2}
		\sched = \{\}\;
		\While{exists a schedulable application}
		{\label{algo:persched_while2}
			$\mathcal{A} = \{\app|\app \text{ is schedulable}\}$\;
			Let \app be the element of $\mathcal{A}$ minimal with respect to
			the lexicographic order $\left(\frac{\tp}{\yieldp},\frac{\wk}{\tio}\right)$\label{algo:persched_min2}\;
			$\sched\!\gets$\iis\!\!$(\sched,\app)$\;
		}
		\eIf{$\maxThrough(\sched) = \frac{T_{\opt}}{T}\cdot \texttt{SE} $\label{algo:persched.se_test}}{
			$\sched_{\opt} \gets \sched$\;
			$T \gets T - (T_{\opt} - \frac{T_{\opt}}{1+\varepsilon})/\lfloor1/\varepsilon\rfloor$
		}{
		      \Return $\sched_{\opt}$
		}
	}
}
\end{algorithm}

\clearpage
\subsection{Complexity analysis}
Finally, in this section we show that our algorithm runs in reasonnable
execution time. We detail theoretical results that allowed us to reduce the
complexity.
We want to show the following result:
\begin{theorem}
	\label{thm:complexity}
Let $\nmax = \left (\frac{\max_k (\wk +\tio)}{\min_k (\wk +\tio)}\right)$,\\
$\persched(K',\varepsilon,\{\app\}_{1\leq k\leq \maxk})$ runs in
\[	
O\left(\left(\left\lceil\frac{1}{\varepsilon}\right\rceil+\left\lceil\frac{\log K'}{\log (1+\varepsilon)}\right\rceil\right)\cdot\maxk^2\left( \nmax + \log K' \right )\right).
\]
\end{theorem}
Some of the complexity results are straightforward. The key results to show are:
\begin{itemize}
	\item The complexity of the tests ``\texttt{while} {\em exists a schedulable
application}'' on lines~\ref{algo:persched_while} and~\ref{algo:persched_while2}
	\item The complexity of computing $\mathcal{A}$ and finding its minimum
	element on line~\ref{algo:persched_min} and~\ref{algo:persched_min2}.
	\item The complexity of \iis
\end{itemize}

To reduce the execution time, we proceed as follows: instead of implementing
the set $\mathcal{A}$, we implement a heap $\tilde{\mathcal{A}}$ that could be
summarized as \[\{\app|\app \text{ is not yet known to {\bf not} be
schedulable}\}\] sorted following the lexicographic order:
$\left(\frac{\tp}{\yieldp},\frac{\wk}{\tio}\right)$. 
Hence, we replace the while loops on lines~\ref{algo:persched_while} and \ref{algo:persched_while2}
by the algorithm snippet described in Algorithm~\ref{algo:snippet}. The idea is to avoid calling \iis after each new inserted instance to know which applications are schedulable.
{
\begin{algorithm}[tbh]
\caption{Schedulability snippet \label{algo:snippet}}
\setcounter{AlgoLine}{10}
	$\tilde{\mathcal{A}} = \cup_{k}\{\app \}$ (sorted by $\left(\frac{\tp}{\yieldp},\frac{\wk}{\tio}\right)$)\;
		\While{$\tilde{\mathcal{A}}\neq \emptyset$}
		{
			Let \app be the minimum element of $\tilde{\mathcal{A}}$\;
			$\tilde{\mathcal{A}}\gets\tilde{\mathcal{A}}\setminus \{\app\} $\;
			Let $\sched'=$\iis\!\!$(\sched,\app)$\;
			\If{$\sched'\neq\sched$}{
				$\sched \gets \sched'$\;
				Insert \app in $\tilde{\mathcal{A}}$ following $\left(\frac{\tp}{\yieldp},\frac{\wk}{\tio}\right)$\;
			}
			
		}

\end{algorithm}
}

We then need to show that they are equivalent, that is:
\begin{itemize}
	\item At all time, the minimum element of $\tilde{\mathcal{A}}$ is minimal
amongst the schedulable applications with respect to the order
$\left(\frac{\tp}{\yieldp},\frac{\wk}{\tio}\right)$ (shown in
Lemma~\ref{lem:min});
	\item If $\tilde{\mathcal{A}} = \emptyset$ then there are no more
schedulable applications (shown in Corollary~\ref{coro:subset}).
\end{itemize}
To show this, it is sufficient to show that (i) at all time, $\mathcal{A}
\subset \tilde{\mathcal{A}}$, and (ii) $\tilde{\mathcal{A}}$ is always sorted
according to $\left(\frac{\tp}{\yieldp},\frac{\wk}{\tio}\right)$.

\begin{definition}[Compact pattern]
  We say that a pattern \[\sched = \cup_{k=1}^{\maxk} \left ( \nper,
  \cup_{i=1}^{\nper} \{\init{k}{i},\initIO{k}{i},\band{k}{}\}\right )\] is
  compact if for all $1 \leq i < \nper$, either $\init{k}{i}+\wk =
  \initIO{k}{i}$, or for all $t \in [\init{k}{i},\initIO{k}{i}]$, $\sum_l \pk[l]
  \band{l}{t} = \maxB$.
\end{definition}
Intuitively, this means that we can only schedule a new instance for all
application \app between \inst{\nper} and \inst{1}.

\begin{lemma}
	\label{lem:compacity}
  At any time during \persched, \sched is compact.
\end{lemma}

\begin{proof}
  For each application, either we use \iisFirst to insert the first instance (so \sched is compact as there is only one instance of an application at this step), either we use
  \iis which inserts an instance just after the last inserted one, which is the definition of being compact. Hence, \sched is compact at any time during \persched.
\end{proof}

\begin{lemma}
\label{lem:return_value}
$\iis(\sched,\app)$ returns \sched, if and only if \app is not schedulable.
\end{lemma}
\begin{proof}
One can easily check that \iis checks the schedulability of \app only between
the last inserted instance of \app and the first instance of \app. Furthermore,
because of the compacity of \sched (Lemma~\ref{lem:compacity}), this is
sufficient to test the overall schedulability.

Then the test is provided by the last condition $\texttt{Dataleft} > 0$.
  \begin{compactitem}
   \item If the condition is false, then the algorithm actually inserts a new
    instance, so it means that one more instance of \app is schedulable.
   \item If the condition is true, it means that we cannot insert a new
    instance after the last inserted one. Because \sched is compact, we cannot
    insert an instance at another place. So if the condition is true, we cannot
    add one more instance of \app in the pattern.
  \end{compactitem}
\end{proof}
\begin{corollary}
	\label{coro:sched}
In Algorithm~\ref{algo:snippet}, an application \app is removed from
$\tilde{\mathcal{A}}$ if and only if it is not schedulable.
\end{corollary}

\begin{lemma}
  If an application is not schedulable at some step, it will not be
  either in the future.
  \label{notschedulable}
\end{lemma}

\begin{proof}
  Let us suppose that \app is not schedulable at some step. In the future, new instances of other applications can be added, thus possibly increasing the total bandwidth used at each instant.
  The total I/O load is non-decreasing during the execution of the algorithm. Thus if for all $i$, we had 
\[
\int_{\init{k}{i}+\wk}^{\initIO{k}{i}-\wk} \min\left (\pk \bandwidth,\maxB - \sum_l \pk[l] \band{l}{t} \right )dt < \io,
\]
  then in the future, with new bandwidths used $\bandprime{l}{t} > \band{l}{t}$, we will still have that for all $i$,
\[
\int_{\init{k}{i}+\wk}^{\initIO{k}{i}-\wk} \min\left (\pk \bandwidth,\maxB - \sum_l \pk[l] \bandprime{l}{t} \right )dt < \io.
\]
\end{proof}

\begin{corollary}
	\label{coro:subset}
At all time, 
\[\mathcal{A} = \{\app|\app \text{ is schedulable}\} \subset \tilde{\mathcal{A}}.\]
\end{corollary}
This is a direct corollary of Corollary~\ref{coro:sched} and Lemma~\ref{notschedulable}
\begin{lemma}
\label{lem:min}
At all time, the minimum element of $\tilde{\mathcal{A}}$ is minimal amongst the
schedulable applications with respect to the order
$\left(\frac{\tp}{\yieldp},\frac{\wk}{\tio}\right)$ (but not necessarily
schedulable).
\end{lemma}
\begin{proof}
First see that $\{\app|\app \text{ is schedulable}\} \subset
\tilde{\mathcal{A}}$.

Furthermore, initially the minimality property is true. Then the set
$\tilde{\mathcal{A}}$ is modified only when a new instance of an application is
added to the pattern. More specifically, only the application that was modified
has its position in $\tilde{\mathcal{A}}$ modified. One can easily verify that
for all other applications, their order with respect to
$\left(\frac{\tp}{\yieldp},\frac{\wk}{\tio}\right)$ has not changed, hence the
set is still sorted.
\end{proof}

This concludes the proof that the snippet is equivalent to the while loops.
With all this we are now able to show timing results for the version of
Algorithm~\ref{algo:persched} that uses Algorithm~\ref{algo:snippet}.

\begin{lemma}
	\label{lem:while2_compl}
The loop on line~\ref{algo:persched_while2} of Algorithm~\ref{algo:persched}
terminates in at most $\lceil1/\varepsilon\rceil$ steps.
\end{lemma}
\begin{proof}
  The stopping criteria on line~\ref{algo:persched.se_test} checks that the
  number of instances did not change when reducing the pattern size. Indeed, by
  definition for a pattern~\sched,
\begin{align*}
\maxThrough(\sched) &= \sum_{k} \pk \yieldp \\
& = \frac{\sum_k \pk \nper \wk}{\period}.
\end{align*}

Denote $\texttt{SE}$ the \maxThrough reached in $T_\opt$ at the end of the
while loop on line~\ref{algo:persched_while} of Algorithm~\ref{algo:persched}.
Let $\maxThrough(\sched)$ be the \maxThrough obtained in
$T_{\opt}/(1+\varepsilon)$. By definition, 
\begin{align*}
\maxThrough(\sched) &< \texttt{SE} &\text{and}\\
\frac{T_{\opt}}{1+\varepsilon}\maxThrough(\sched) &< T_{\opt}\texttt{SE}. &
\end{align*}

Necessarily, after at most $\lceil1/\varepsilon\rceil$ iterations,
Algorithm~\ref{algo:persched} exits the loop on line~\ref{algo:persched_while2}.

\end{proof}


\begin{proof}[Proof of Theorem~\ref{thm:complexity}]
There are $\lfloor m \rfloor$ pattern sizes tried where $\minperiod \cdot
(1+\varepsilon)^m = \maxperiod$ in the main ``while'' loop
(line~\ref{algo:while_main}), that is
\[m = \frac{\log \maxperiod - \log \minperiod}{ \log (1+\varepsilon)}= \frac{\log K'}{ \log (1+\varepsilon)}.\]
Furthermore, we have seen (Lemma~\ref{lem:while2_compl}) that there are a
maximum of $\lceil1/\varepsilon\rceil$ pattern sizes tried of the second loop
(line~\ref{algo:while_main2}).

For each pattern size tried, the cost is dominated by the complexity of
Algorithm~\ref{algo:snippet}. Let us compute this complexity.
\begin{itemize}
	\item The construction of $\tilde{\mathcal{A}}$ is done in
$O(\maxk\log\maxk)$.
	\item In sum, each application can be inserted a maximum of \nmax times in
$\tilde{\mathcal{A}}$ (maximum number of instances in any pattern), that is the
total of all insertions has a complexity of $O(\maxk\log\maxk \nmax)$.
\end{itemize}

We are now interested by the complexity of the different calls to \iis.

First one can see that we only call \iisFirst \maxk times, and in particular
they correspond to the first \maxk calls of \iis.
Indeed, we always choose to insert a new instance of the application that has
the largest current slowdown. The slowdown is infinite for all applications at
the beginning, until their first instance is inserted (or they are
removed from $\tilde{\mathcal{A}}$) when it becomes finite, meaning that the \maxk first insertions will be
the first instance of all applications.

During the $k$-th call, for $1 \leq k \leq \maxk$, there will be $n = 2(k-1)+2$
events (2 for each previously inserted instances and the two bounds on the
pattern), meaning that the complexity of \iisFirst will be $O(n\log{n})$
(because of the sorting of the bandwidths available by non-increasing order to
choose the intervals to use). So overall, the \maxk first calls have a
complexity of $O(\maxk^2\log{\maxk})$.

Furthermore, to understand the complexity of the remaining calls to \iis we are
going to look at the end result. 
In the end there is a maximum of \nmax instance
of each applications, that is a maximum of $2\nmax \maxk$ events.
For all application \app, for all instance \inst{i}{k}, $1<i\leq\numk$, the only
events considered in \iis when scheduling \inst{i}{k} were the ones between the
end of $\init{i}{k}+\wk$ and $\init{i}{k+1}$. Indeed, since the schedule has
been able to schedule \io, \iis will exit the while loop on
line~\ref{iis.line.while}.
Finally, one can see that the events considered for all instances of an
application partition the pattern without overlapping. Furthermore, \iis has a
linear complexity in the number of events considered. Hence a total complexity
by application of $O(\nmax \maxk)$. Finally, we have \maxk applications, the
overall time spent in \iis for inserting new instances is $O(\maxk^2 \nmax)$.

Hence, with the number of different pattern tried, we obtain a complexity of
\[O\left ( \left (\lceil m \rceil + \left\lceil\frac{1}{\varepsilon}\right\rceil\right ) \left (\maxk^2\log{\maxk} + \maxk^2\nmax \right ) \right).\] 


\end{proof}

Note that in practice, both $K'$ and \maxk are small ($\approx10$), and
$\varepsilon$ is close to 0, hence making the complexity $O\left
(\frac{\nmax}{\varepsilon} \right )$.

\subsection{High-level implementation, proof of concept}
	\label{sec.implem}
We envision the implementation of this periodic scheduler to take place at two
levels:

1) The job scheduler would know the applications profile (using solutions such
as Omnisc'IO~\cite{dorier2014omnisc}). Using profiles it would be in charge of
computing a periodic pattern every time an application enters or leaves the
system.

2) Application-side I/O management strategies (such
as~\cite{zhang2012opportunistic,lofstead2010managing,tessier2016topology}) then would be responsible
to ensure the correct transfer of I/O at the right time by limiting the
bandwidth used by nodes that transfer I/O. The start and end time for each I/O as well as the used bandwidth are described in input files. 
	\section{Evaluation and model validation}
	\label{sec.simu}

{\em Note that the data used for this section and the scripts to generate the
figures are available at \url{https://github.com/vlefevre/IO-scheduling-simu}.}

In this section, we (i) assess the efficiency of our algorithm by comparing it
to a recent dynamic framework~\cite{gainaru2015scheduling}, and (ii) validate
our model by comparing theoretical performance (as obtained by the simulations)
to actual performance on a real system.
	
We perform the evaluation in three steps: first we simulate behavior of
applications and input them into our model to estimate both \minUserCong and
\maxThrough of our algorithm (Section~\ref{sec.simu.results}) and evaluate these
cases on an actual machine to confirm the validity of our model. Finally, in
Section~\ref{sec.simu.tmax} we confirm the intuitions introduced in
Section~\ref{sec.algo} to determine the parameters used by \persched.

	\subsection{Experimental Setup}
	\label{sec.simu.setup}

The platform available for experimentation is Jupiter at Mellanox, Inc. To be
able to verify our model, we use it to instantiate our platform model.
Jupiter is a Dell PowerEdge R720xd/R720 32-node cluster using Intel Sandy Bridge
CPUs. Each node has dual Intel Xeon 10-core CPUs running at 2.80 GHz, 25 MB of
L3, 256 KB unified L2 and a separate L1 cache for data and instructions, each 32
KB in size. The system has a total of 64GB DDR3 RDIMMs running at 1.6 GHz per
node. Jupiter uses Mellanox ConnectX-3 FDR 56Gb/s InfiniBand and Ethernet VPI
adapters and Mellanox SwitchX SX6036 36-Port 56Gb/s FDR VPI InfiniBand switches. 

We measured the different bandwidths of the machine and obtained $b = 0.01$GB/s
and $B = 3$GB/s. Therefore, when 300 cores transfer at full speed (less than
half of the 640 available cores), congestion occurs.

\paragraph*{Implementation of scheduler on Jupiter}

We simulate the existence of such a scheduler by computing beforehand the I/O
pattern for each application and feeding it as input files. The experiments
require a way to control the exact moment when all applications perform I/O, use
the CPU or stay idle waiting to start their I/O. For this purpose, we modified
the IOR benchmark~\cite{iorbench} to read the input files that provide the start and end time for each I/O transfer as well as the bandwidth used. Our scheduler generates one such file for each application. 
Each IOR instance represents one application whose I/O pattern is described in
one of the generated scheduling files.
The IOR benchmark is split in different sets of processes running independently
on different nodes, where each set represents a different application. One
separate process acts as the scheduler and receives I/O requests for all groups
in IOR. Since we are interested in modeling the I/O delays due to congestion or
scheduler imposed delays, the modified IOR benchmarks do not use inter-processor
communications.

We made experiments on our IOR benchmark and compared the results
between periodic and online schedulers as well as with the performance of the
original IOR benchmark without any extra scheduler.

\subsection{Applications and scenarios}
\label{sec.simu.app}

In the literature, there are many examples of periodic applications.
Carns et
al.~\cite{carns200924} observed with Darshan the periodicity of four different
applications (MADBench2~\cite{carter2005performance}, Chombo I/O
benchmark~\cite{Chombo}, S3D IO~\cite{sankaran2006direct} and
HOMME~\cite{nair2007petascale}). Furthermore, in our previous
work~\cite{gainaru2015scheduling} we were able to verify the periodicity of
gyrokinetic toroidal code (GTC)~\cite{ethier2012petascale},
Enzo~\cite{enzoenzo}, HACC application~\cite{habib2012universe} and
CM1~\cite{bryan2002benchmark}.

Unfortunately, few documents give the actual values for \wk, \io and \pk. Liu et
al.~\cite{Liu12onthe} provide different periodic patterns of four scientific
applications: PlasmaPhysics, Turbulence1, Astrophysics and Turbulence2. They
were also the top four write-intensive jobs run on Intrepid in 2011. We chose
the most I/O intensive patterns for all applications (as they are the most
likely to create I/O congestion). We present these results in
Table~\ref{table.apps}. Note that to scale those values to our system, we
divided the number of processors \pk by 64, hence increasing \wk by 64. The I/O
volume stays constant.

\begin{table}[t]
\begin{center}
\begin{tabular}{|lr|r|r|r|}
\hline
	\app & &\wk (s) & \io (GB) & \pk \\
	\hline
	Turbulence1 &(T1)& 70 & 128.2 & 32,768 \\
	\hline
	Turbulence2 &(T2) & 1.2 & 235.8 & 4,096 \\
	\hline
	AstroPhysics &(AP) & 240 & 423.4 & 8,192 \\
	\hline
	PlasmaPhysics &(PP) & 7554 & 34304 & 32,768 \\
	\hline
\end{tabular}
\caption{Details of each application.}
\label{table.apps}
\end{center}
\end{table}

To compare our strategy, we tried all possible combinations of those
applications such that the number of nodes used equals 640. That is a total of
ten different scenarios that we report in Table~\ref{table.sets}.

\begin{table}[t]
\begin{center}
\setlength\tabcolsep{2.5pt}
\begin{tabular}{|c|c|c|c|c|}
\hline
Set \#& ~T1~ & ~T2~ & ~AP~ & ~PP~\\
\hline
1 & 0 & {\bf 10} & 0 & 0 \\
2 & 0 & {\bf 8} & {\bf 1} & 0 \\
3 & 0 & {\bf 6} & {\bf 2} & 0 \\
4 & 0 & {\bf 4} & {\bf 3} & 0 \\
5 & 0 & {\bf 2} & 0 & {\bf 1} \\
6 & 0 & {\bf 2} & {\bf 4} & 0 \\
7 & {\bf 1} & {\bf 2} & 0 & 0 \\
8 & 0 & 0 & {\bf 1} & {\bf 1} \\
9 & 0 & 0 & {\bf 5} & 0 \\
10 & {\bf 1} & 0 & {\bf 1} & 0 \\
\hline
\end{tabular}
\caption{Number of applications of each type launched at the same time for each experiment scenario.}
\label{table.sets}
\end{center}
\end{table}

\subsection{Baseline and evaluation of existing degradation}

We ran all scenarios on Jupiter without any additional scheduler. In all tested
scenarios congestion occurred and decreased the visible bandwidth used by each
applications as well as significantly increased the total execution time.
We present in Table~\ref{table.online} the average I/O bandwidth slowdown due to
congestion for the most representative scenarios together with the corresponding
values for \maxThrough. Depending on the IO transfers per computation ratio of
each application as well as how the transfers of multiple applications overlap,
the slowdown in the perceived bandwidth ranges between 25\% to 65\%. 

\begin{table}
\begin{center}
\begin{tabular}{|c|c|c|c|}
\hline
Set \#&  Application & BW slowdown & \maxThrough  \\
\hline
1 &  Turbulence 2 &  65.72\% & 0.064561 \\\hline
2 &  Turbulence 2 &  63.93\% & 0.250105 \\ 
& AstroPhysics & 38.12\% &  \\\hline
3 &  Turbulence 2 & 56.92\% & 0.439038\\ 
   & AstroPhysics & 30.21\% &  \\\hline
4 &  Turbulence 2 & 34.9\% & 0.610826 \\
   & AstroPhysics & 24.92\% &  \\\hline
6 & Turbulence 2 & 34.67\% & 0.621977 \\ 
& AstroPhysics & 52.06\% &  \\\hline
10 &   Turbulence 1 & 11.79\% & 0.98547 \\ 
	& AstroPhysics & 21.08\% &   \\\hline
\end{tabular}
\end{center}
\caption{Bandwidth slowdown, performance and application slowdown for each set
of experiments}
\label{table.online}
\end{table}

Interestingly, set 1 presents the worst degradation. This scenario is running
concurrently ten times the same application, which means that the I/O for all
applications are executed almost at the same time (depending on the small
differences in CPU execution time between nodes). This scenario could correspond
to coordinated checkpoints for an application running on the entire system. The
degradation in the perceived bandwidth can be as high as 65\% which
considerably increases the time to save a checkpoint. The use of I/O schedulers
can decrease this cost, making the entire process more efficient.

\subsection{Comparison to online algorithms}
\label{sec.simu.results}

In this subsection, we present the results obtained by running \persched and the
online heuristics from our previous work~\cite{gainaru2015scheduling}. Because
in~\cite{gainaru2015scheduling} we had different heuristics to optimize either
\minUserCong or \maxThrough, in this work, the \minUserCong and
\maxThrough presented are the best reached by {\em any} of those heuristics. 
This means that \emph{there are no online solution able to reach them both at
the same time}! We show that even in this scenario, our algorithm outperforms
these heuristics {\em for both optimization problems}!

\persched takes as input a list of applications, as well as the parameters,
presented in Section~\ref{sec.algo}, $K'=\frac{\maxperiod}{\minperiod}$,
$\varepsilon$. All scenarios were tested with $K'=10$ and $\varepsilon=0.01$.

\paragraph*{Simulation results}
We present in Table~\ref{table.results} all evaluation results.
The results obtained by running Algorithm~\ref{algo:persched} are called
\persched. To go further in our evaluation, we also look for the best
\minUserCong obtainable with our pattern (we do so by changing
line~\ref{algo:se} of \persched). We call this result \emph{min \minUserCong} in
Table~\ref{table.results}. This allows us to estimate how far the \minUserCong
that we obtain is from what we can do. Furthermore, we can compute an upper
bound to \maxThrough by replacing \yieldy by \tp in Equation~\eqref{eq:syseff}:
\begin{equation}
\label{eq:upperbound}
\text{Upper bound} =\frac{1}{N} \sum_{k=1}^\maxk \frac{\pk \wk}{\wk+\tio[k]} .
\end{equation}

\begin{table}[h]
\makegapedcells
\setcellgapes{2pt}
\setlength\tabcolsep{2.5pt}
\begin{center}
\resizebox{\columnwidth}{!}{
\begin{tabular}{|c||cc||cc|cc|}
	\hline
	\multirow{2}{*}{Set} & Min & Upper bound & \multicolumn{2}{c|}{\persched} & \multicolumn{2}{c|}{Online}\\
	\cline{4-7}
			&	\minUserCong	&	\maxThroughShort	&	\minUserCong	&	\maxThroughShort	&	\minUserCong	&	\maxThroughShort	\\
		\hline													
		1	&	1.777	&	0.172	&	1.896	& 	0.0973	&	2.091	&	0.0825	\\
		\hline													
		2	&	1.422	&	0.334	&	1.429	&	0.290	&	1.658	&	0.271	\\
		\hline													
		3	&	1.079	&	0.495	&	1.087	&	0.480	&	1.291	&	0.442	\\
		\hline													
		4	&	1.014	&	0.656	&	1.014	&	0.647	&	1.029	&	0.640	\\
		\hline													
		5	&	1.010	&	0.816	&	1.024	&	0.815	&	1.039	&	0.810	\\
		\hline
		6	&	1.005	&	0.818	&	1.005	&	0.814	&	1.035	&	0.761	\\
		\hline													
		7	&	1.007	&	0.827	&	1.007	&	0.824	&	1.012	&	0.818	\\
		\hline													
		8	&	1.005	&	0.977	&	1.005	&	0.976	&	1.005	&	0.976	\\
		\hline													
		9	&	1.000	&	0.979	&	1.000	&	0.979	&	1.004	&	0.978	\\
		\hline											
		10	&	1.009	&	0.988	&	1.009	&	0.986	&	1.015	&	0.985	\\
		\hline															

\end{tabular}
}
\end{center}
\caption{Best \minUserCong and \maxThrough for our periodic heuristic and online heuristics.}
\label{table.results}
\end{table}

The first noticeable result is that \persched almost always outperfoms (when
it does not, matches) both the
\minUserCong and \maxThrough attainable by the online scheduling algorithms! 
This is particularly impressive as these objectives are not obtained by the same
online algorithms (hence conjointly), contrarily to the \persched result.

While the gain is minimal (from 0 to 3\%, except \maxThrough increased by 7\% for case 4) when little congestion occurs
(cases 4 to 10), the gain is between 9\% and  16\% for \minUserCong
and between 7\% and 18\% for \maxThrough when congestion occurs (cases 1, 2, 3)!

The value of $\varepsilon$ has been chosen so that the computation stays short. It seems to be a good
compromise as the results are good and the execution times vary from 4 ms (case 10) to 1.8s (case 5)
using a Intel Core I7-6700Q. Note that the algorithm is easily parallelizable, as each iteration of the loop
is independent. Thus it may be worth considering a smaller value of $\varepsilon$, but there will
be no big improvement on the results.

\paragraph*{Model validation through experimental evaluation}

We used the modified IOR benchmark to reproduce the behavior of applications
running on HPC systems and analyze the benefits of I/O schedulers. We made
experiments on the 640 cores of the Jupiter system. Additionally to the results
from both periodic and online heuristics, we present the performance of the
system with no additional I/O scheduler.

\begin{figure}[tbh]
  \centering
	\subfloat[\maxThrough / Upper bound \maxThrough]{\includegraphics[width=0.5\columnwidth]{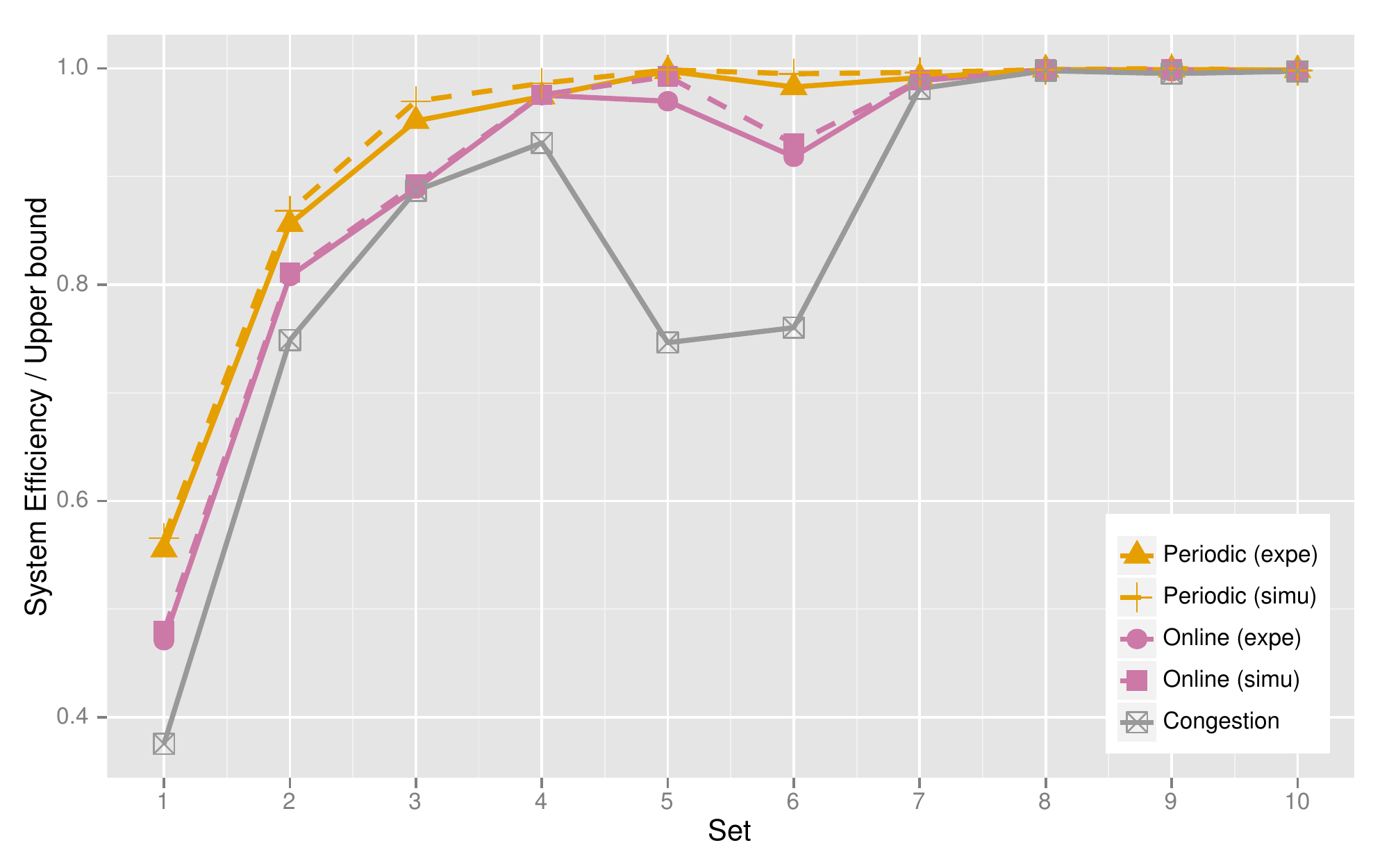}}
	\subfloat[\minUserCong]{\includegraphics[width=0.5\columnwidth]{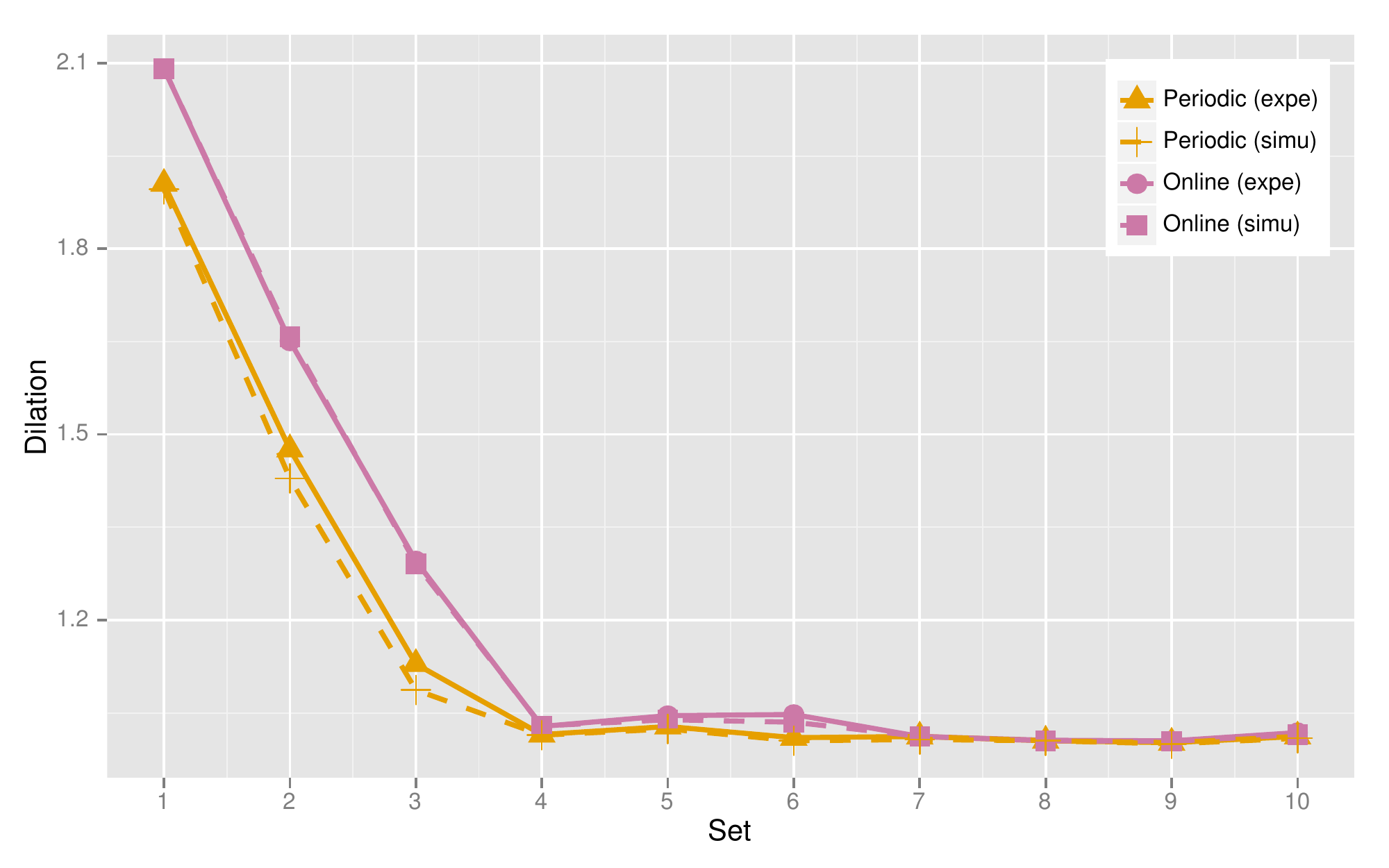}}
\caption{Performance for both experimental evaluation and theoretical
(simulated) results. The performance estimated by our model is accurate within
3.8\% for periodic schedules and 2.3\% for online schedules.}
 \label{fig.vsonline}
\end{figure}

Figure~\ref{fig.vsonline} shows the \maxThrough (normalized using the upper
bound in Table~\ref{table.results}) and \minUserCong when using the periodic
scheduler in comparison with the online scheduler. For the system efficiency the
upper limit and the results when applications are running without any scheduler
are also shown. As observed in the previous section, the periodic scheduler
gives better or similar results to the best solutions that can be returned by
the online ones, in some cases increasing the system performance by 18\% and
the dilation by 13\%. When we compare to the current strategy on Jupiter, the
\maxThrough reach 48\%! In addition, the periodic scheduler has the benefit of
not requiring a global view of the execution of the applications at every moment
of time (by opposition to the online scheduler).

Finally, a key information from those results is the precision of our
model introduced in Section~\ref{sec.model}. The theoretical results (based on
the model) are within 3\% of the experimental results!

{\em This observation is key in launching more thorough evaluation via extensive
simulations and is critical in the experimentation of novel periodic scheduling
strategies.}

	\subsection{Discussion on finding the best pattern size}
	\label{sec.simu.tmax}
	
The core of our algorithm is a search of the best pattern size via an exponential
growth of the pattern size until \maxperiod. As stated in Section~\ref{sec.algo}, the
intuition of the exponential growth is that the larger the pattern size, the less
needed the precision for the pattern size as it might be easier to fit many instances
of each application. On the contrary, we expect that for small pattern sizes finding
the right one might be a precision job.

\begin{figure}[tbh]
\centering
	\includegraphics[width=\columnwidth]{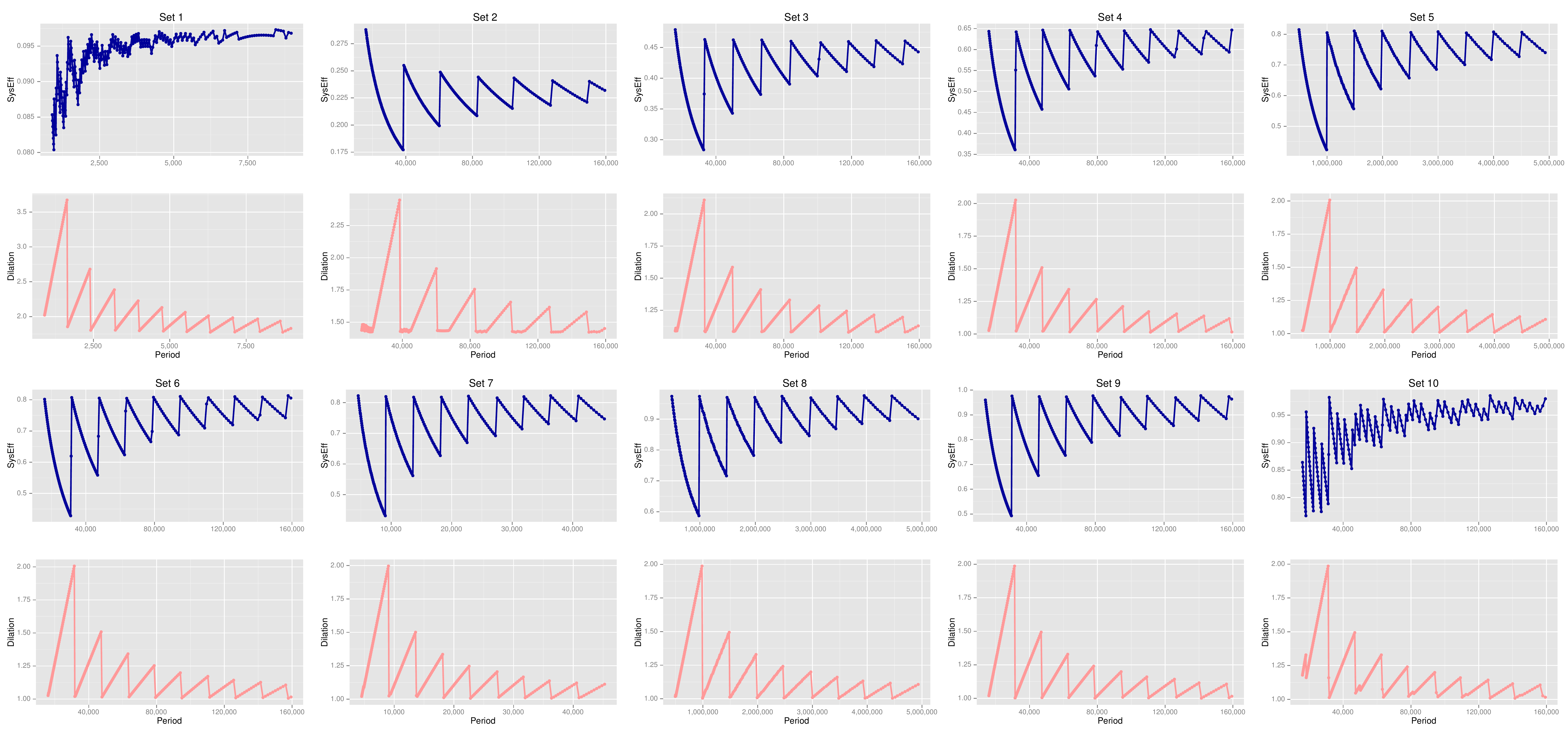}
\caption{Evolution of \maxThrough and \minUserCong when the pattern size increases.}
\label{fig.evolutionT}
\end{figure}
We verify this experimentally and plot on Figure~\ref{fig.evolutionT} the
\maxThrough and \minUserCong found by our algorithm as a function of the pattern size
\period.

Finally, the last information to determine to tweak \persched is the value of
\maxperiod. Remember that we denote $K' = \maxperiod / \minperiod$.

\begin{table}
\begin{center}
\begin{tabular}{c|c|c}
Set & \ninst & \nmax \\
\hline
1 & 11 & 1.00\\
2 & 25 & 35.2\\
3 & 33 & 35.2\\
4 & 247  & 35.2\\
5 & 1086 & 1110\\
\end{tabular}
\hspace{0.5cm}
\begin{tabular}{c|c|c}
Set & \ninst & \nmax \\
\hline
6 & 353 & 35.2\\
7 & 81 & 10.2\\
8 & 251 & 31.5\\
9 & 9 & 1.00\\
10 & 28 & 3.47\\
\end{tabular}
\end{center}
\caption{Maximum number of instances (\ninst) per application, ratio between longest and shortest application (\nmax) in the solution returned by \persched.}
\label{table.instances}
\end{table}

To be able to get an estimate of the pattern size returned by \persched, we
provide in Table~\ref{table.instances} (i) the maximum number of instances
\ninst of any application, and (ii) the ratio 
$\nmax =\frac{\max_k \left (\wk+\tio \right)}{\min_k \left (\wk+\tio\right)}$.
Together along with the fact that the \minUserCong (Table~\ref{table.results})
is always below 2 they give a rough idea of $K'$ ($\approx
\frac{\ninst}{\nmax}$). It is sometimes close to 1, meaning that a small value of $K'$ can be sufficient, but choosing $K' \approx 10$ is necessary in the general case.

We then want to verify the cost of under-estimating \maxperiod. For this
evaluation all runs were done up to $K'=100$ with $\varepsilon=0.01$.
Denote $\maxThrough(K')$ (resp. $\minUserCong(K')$) the maximum \maxThrough
(resp. corresponding \minUserCong) obtained when running \persched with $K'$. We plot
their normalized version that is:
\[
\frac{\maxThrough(K')}{\maxThrough(100)} \left (\text{resp. } \frac{\minUserCong(K')}{\minUserCong(100)} \right )
\]
on Figure~\ref{fig.tmax}.
\begin{figure}[tbh]
\centering
\includegraphics[width=0.8\columnwidth]{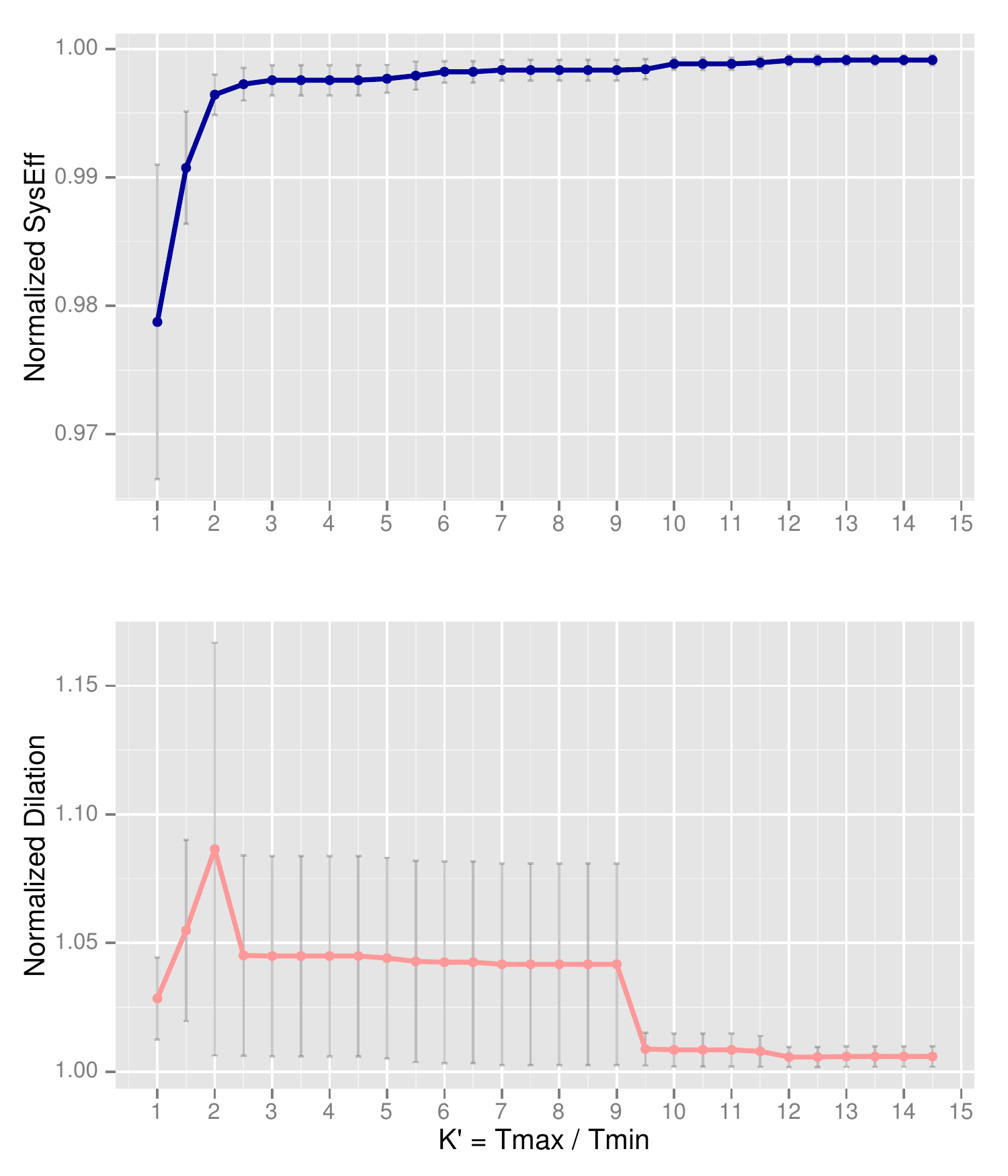}
\caption{Normalized system efficiency and dilation obtained by
Algorithm~\ref{algo:persched} averaged on all 10 sets as a function of $K'$ (with Standard Error bars).}
\label{fig.tmax}
\end{figure}
The main noticeable information is that the convergence is very fast: when
$K'=3$, the average \maxThrough is within $0.3\%$ of $\maxThrough(100)$, but the
corresponding average \minUserCong is $5\%$ higher than $\minUserCong(100)$. If
we go to $K' = 10$ then we have a \maxThrough of $0.1\%$ of $\maxThrough(100)$
and a \minUserCong within $1\%$ of $\minUserCong(100)$!
Hence validating that choosing $K'=10$ is sufficient.

	\section{Related Work}
	\label{sec.related}

Performance variability due to resource sharing can significantly detract from the suitability of a given architecture for a workload as well as from the overall performance realized by parallel workloads~\cite{hpc3}. Over the last decade there have been studies to analyze the sources of performance degradation and several solutions have been proposed. 
In this section, we first detail some of the existing work that copes with I/O congestion and then we present some of the theoretical literature that is similar to our \periodic problem.

The storage I/O stack of current HPC systems has been increasingly identified as a performance bottleneck. Significant improvements in both hardware and software need to be addressed to overcome oncoming scalability challenges. The study in~\cite{newio1} argues for making data staging coordination driven by generic cross-layer mechanisms that enable global optimizations by enforcing local decisions at node granularity at individual stack layers.

While many other studies suggest that I/O congestion is one of the main problems for future scale 
platforms~\cite{biswas2007petascale,lofstead2013insights}, few papers focus on finding a solution 
at the platform level. Some paper consider application-side I/O 
management and transformation (using aggregate nodes, compression
etc)~\cite{zhang2012opportunistic,lofstead2010managing,tessier2016topology}. We
consider those work to be orthogonal to our work and able to work jointly.
Recently, numerous works focus on using machine learning for auto tuning and
performance studies~\cite{behzad2013taming,kumar2013characterization}. However
these solution also work at the application level for IO-scheduling and do not
have a global view of the I/O requirements of the system and they need to be
supported by a platform level I/O management for better results.

Some paper consider the use of burst buffers to reduce I/O congestion by
delaying accesses to the file storage, as they found that congestion occurs on a
short period of time and the bandwidth to the storage system is often
underutilized~\cite{Liu12onthe}. However, the computation power tends
to increase faster than the I/O bandwidth, which may cause the bandwidth to be saturated more often and thus decreasing the efficiency of burst buffers. \cite{kougkas2016leveraging}~presents a dynamic
I/O scheduling at the application level using burst buffers to stage I/O and to allow computations to continue uninterrupted. They design different strategies to mitigate I/O interference, including
partitioning the PFS, which reduces the effective bandwidth non-linearly. However, the strategies are basically designed for only 2 applications and their heuristics does not take into account
the characteristics of the applications to better optimize the scheduling.

The study from~\cite{vios} offers ways of isolating the performance experienced by applications of one operating system from variations in the I/O request stream characteristics of applications of other operating systems. While their solution cannot be applied to HPC systems, the study offers a way of controlling the coarse grain allocation of disk time to the different operating system instances as well as determining the fine-grain interleaving of requests from the corresponding operating systems to the storage system.

Closer to this work, online schedulers for HPC systems were developed such as
our previous work~\cite{gainaru2015scheduling}, the study by Zhou et al~\cite{newio2}, and a solution proposed by Dorier
et al~\cite{matthieu}. In~\cite{matthieu}, the authors investigate the
interference of two applications and analyze the benefits of interrupting or
delaying either one in order to avoid congestion. Unfortunately their approach
cannot be used for more than two applications. Another main difference with our previous work is the light-weight approach of this study where the computation is only done once.

Our previous study~\cite{gainaru2015scheduling} is more
general by offering a range of options to schedule each I/O performed by an
application.
Similarly, the work from~\cite{newio2} also utilizes a global job scheduler
to mitigate I/O congestion by monitoring and controlling
jobs’ I/O operations on the fly. Unlike online solutions, this paper focuses on a decentralized approach where the scheduler is
integrated into the job scheduler and computes ahead of time, thus overcoming the need to monitor the I/O traffic of each application at every moment of time.


As a scheduling problem, our problem is somewhat close to the cyclic scheduling
problem (we refer to Hanen and Munier~\cite{hanen1993cyclic} for a survey),
namely there are given a set of activities with time dependency between
consecutive tasks stored in a DAG that should be executed on $p$ processors. 
The main difference is that in cyclic scheduling there is no consideration of a
constant time between the end of the previous instance and the next instance.

	\section{Conclusion}
	\label{sec.conclusion}

Performance variation due to resource sharing in HPC systems is a reality and I/O congestion is currently one of the main causes of degradation. Current storage systems are unable to keep up with the amount of data handled by all applications running on an HPC system, either during their computation or when taking checkpoints. 
In this document we have presented a novel I/O scheduling technique that offers a decentralized solution for minimizing the congestion due to application interference. Our method takes advantage of the periodic nature of HPC applications by allowing the job scheduler to pre-define each application's I/O behavior for their entire execution. Recent studies~\cite{dorier2014omnisc} have shown that HPC applications have predictable I/O patterns even when they are not completely periodic, thus we believe our solution is general enough to easily include the large majority of HPC applications.

We conducted simulations for different scenarios and made experiments to validate our results. Decentralized solutions are able to improve both total system efficiency and application dilation compared to dynamic state-of-the-art schedulers. Moreover, they do not require a constant daemon capable of monitoring the state of all applications, nor do they require a change in the current I/O stack. One particularly interesting result is for scenario~1 with 10 identical periodic
behaviors (such as what can be observed with periodic checkpointing for
fault-tolerance). In this case the periodic scheduler shows a 30\% improvement
in \maxThrough. Thus, system wide applications taking global checkpoints could
benefit from such a strategy. 

{\em Future work:} we believe this work is the initialization of a new set of
techniques to deal with the I/O requirements of HPC system. In particular,
by showing the efficiency of the periodic technique on simple pattern, we expect
to open a door to multiple extensions. We give here some examples that we will
consider in the future.
The next natural directions is to take more complicated periodic shapes for
applications (an instance could be composed of sub-instances) as well as
different point of entry inside the job scheduler (multiple IO nodes). 
This would be modifying the \iis procedure and we expect that this should work
well as well. 
Another future step would be to study how variability in the compute or I/O
volumes impact a periodic schedule or the impact of non periodic applications.
Finally we plan to model burst buffers and to show how to use them conjointly
with periodic schedules.

Our method is used for minimizing the congestion caused by concurrent I/O
accesses. However, the methodology and concepts are general and can be applied
to any resource sharing problem. We will continue to investigate the causes for
performance degradation in HPC applications and adapt our findings to each case.

\section*{Acknowledgement}
Part of this work was done when Guillaume Aupy and Valentin Le Fèvre were at
Vanderbilt University. The authors would like to thank Anne Benoit and Yves
Robert for helpful discussions.

\bibliographystyle{abbrv}
\bibliography{biblio}

\end{document}